\documentclass{article}

\usepackage{makeidx}  % allows for indexgeneration
\usepackage{amsmath,amssymb,amsfonts,latexsym,epsfig} 
\usepackage{psfrag, xspace}
\usepackage{hyperref}
\usepackage{framed}
\usepackage{bm}

\usepackage{ccfonts}

\usepackage{fullpage}

\usepackage[numbers,sort]{natbib}
\usepackage{notations}

\usepackage{hyperref}
\newtheorem{theorem}{Theorem}
\newtheorem{lemma}[theorem]{Lemma}

\newtheorem{proposition}[theorem]{Proposition}

\newtheorem{remark}[theorem]{Remark}
\newtheorem{problem}{Problem}

\newcommand{\M}{\ensuremath{\mathsf{M}}\xspace}
\newcommand{\x}{\ensuremath{\mathbf{x}}\xspace}

\begin{document}

\title{Univariate real root isolation in an extension field}

\author{
  Adam Strzebonski%
  \thanks{Wolfram Research Inc.,100 Trade Centre Drive, Champaign, IL 61820, U.S.A.
    Email: \texttt{adams (AT) wolfram.com}}
   \and
   Elias P.~Tsigaridas%
   \thanks{Computer Science Department, Aarhus University, Denmark.
     Email: \texttt{elias(at)cs.au.dk}}
}
\date{}
\maketitle

\begin{abstract}
  We present algorithmic, complexity and implementation results for
  the problem of isolating the real roots of a univariate polynomial
  in $B_{\alpha} \in L[y]$, where $L=\QQ(\alpha)$ is a simple
  algebraic extension of the rational numbers.
  We revisit two approaches for the problem. In the first
  approach, using resultant computations, we perform a reduction to a
  polynomial with integer coefficients and we deduce a bound of
  $\sOB(N^{10})$ for isolating the real roots of $B_{\alpha}$, where
  $N$ is an upper bound on all the quantities (degree and bitsize) of
  the input polynomials.
  In the second approach we isolate the real roots working directly on
  the polynomial of the input.  We compute improved separation bounds
  for the roots and we prove that they are optimal, under mild
  assumptions.  For isolating the real roots we consider a modified Sturm
  algorithm, and a modified version of \func{descartes}' algorithm
  introduced by Sagraloff.  For the former we prove a complexity bound
  of $\sOB(N^8)$ and for the latter a bound of $\sOB(N^{7})$.
  We implemented the algorithms in \func{C} as part of the core
  library of \mathematica and we illustrate their efficiency over
  various data sets.
  Finally, we present complexity results for the general case of the
  first approach, where the coefficients belong to multiple
  extensions.
\end{abstract}

% \vspace{0.9mm}
% \noindent
% {\bf Categories and Subject Descriptors:}
% F.2 [Theory of Computation]: Analysis of Algorithms and Problem Complexity;
% I.1 [Computing Methodology]: Symbolic and algebraic manipulation: Algorithms

{\bf Keywords} real root isolation, algebraic polynomial,
field extension, separation bounds, Sturm, Descartes' rule of sign

\section{Introduction}

Real root isolation is a very important problem in computational
mathematics.  Many algorithms are known for isolating the real roots
of a polynomial with integer or rational coefficients that are either
based solely on operations with rational numbers,
\cite{MoVrYa02,Dav:TR:85,ESY:descartes,RouZim:solve:03} and references
therein, or they follow a numerical, but certified approach,
\cite{Pan02jsc,Sch82} and references therein.  In this paper we
consider a variation of the problem in which the coefficients of the
polynomial are polynomial functions of a real algebraic number, that
is they belong to a simple algebraic extension of the rationals. 
%To be more specific, we consider the following problem:
%%
\begin{framed}
  \begin{problem}
    \label{prob:main}
    Let $\alpha$ be a real algebraic number with isolating
    interval representation $\alpha \cong ( A, \interval{I})$, where
    $A=\sum_{i=0}^{m}{ a_i\, x^i}$, $\interval {I} = [ \rat{a}_1
    ,\rat{a_2}]$, $\rat{ a}_{1,2} \in \QQ$ and $\dg{ A} = m$ and
    $\bitsize{ A} = \tau$.
    Let $B_{\alpha} = \sum_{i=0}^{n}{ b_i(\alpha) \, y^i} \in
    \ZZ(\alpha)[y]$ be square-free, where $b_i(x) =
    \sum_{j=0}^{\eta_i}{ c_{i,j} \, x^j} \in \ZZ[x]$,
    $\bitsize{c_{i,j}} \leq \sigma$, and $\eta_i < m$, for $0 \leq i
    \leq d$.
    What is the Boolean complexity of isolating the real roots of $B_{\alpha}$?
  \end{problem}
\end{framed}
Rump \cite{rump-sigsam-1977}, see also \cite{rump76:algnum}, presented
an algorithm for the problem that is an extension of Collins and Loos
\cite{cl-issac-1976} algorithm for integral  polynomials.
Johnson \cite{Johnson-phd-91} presented and compared various
algorithms for Problem~\ref{prob:main}. He considered a norm based
algorithm that reduces the problem to root isolation of integral
polynomial (this is the approach that we consider in
Sec.~\ref{sec:RIC}) and extended three algorithms used for integral
polynomials, i.e. Sturm (we present it in Sec.~\ref{sec:sturm}), the
algorithm based on derivative sequence and Rolle's theorem
\cite{cl-issac-1976}, and the algorithm based on Descartes' rule of
sign \cite{ColAkr:descartes:76} (we present a modified version in
Sec.~\ref{Sec:BMD}).
Johnson and Krandick \cite{jk-issac-1997} modified the latter and
managed to replace exact arithmetic, when possible, with certified
floating point operations; a novelty that speeds up considerably the
computations. Along the same lines, Rouillier and Zimmermann
\cite{RouZim:solve:03} presented an optimal in terms of memory used
algorithm for integral polynomials that exploits adaptive
multiprecision techniques that could be used for
Problem~\ref{prob:main}, if we approximate the real algebraic number
up to a sufficient precision.
In a series of works \cite{eigenwillig-phd,EigenEtal:bitstream:05,ms-jsc-2010} a
bitstream version of Descartes' algorithm was introduced. The
coefficients of the input polynomial are considered to be real numbers
that we can approximate up to arbitrary precision.
We use the most recent version of this approach, which is due to Sagraloff
\cite{s-arxiv-isol-10}, to tackle Problem~\ref{prob:main}.
Last but not least, let us also mention the numerical algorithms due
to Pan \cite{Pan02jsc} and Sch\"onhage \cite{Sch82}, that could be also used 
if approximate $\alpha$ in our problem up to a sufficient precision.

Rioboo \cite{rioboo03:tfan:jsc} considered various symbolic algorithms
for operations with real algebraic numbers, based on quasi Sylvester
sequences. These algorithms could be used for
Problem~\ref{prob:main}, and they are closely connected with the Sturm
algorithm that we present (Sec.~\ref{sec:sturm}).  However, we
use different subalgorithms for sign evaluations and solving
polynomials. The focus in \cite{rioboo03:tfan:jsc} is on efficient
implementation of the real closure in \func{axiom}.

Problem~\ref{prob:main} is closely related to real root isolation of
of triangular systems and regular chains.
In \cite{cgy-jsc-2009,xz-cmwa-2006,xy-jsc-2002,lhlp-snc-2005}
algorithms and implementations are presented for isolating the real
roots of triangular polynomial systems, based on interval arithmetic
and the so-called {\em sleeve} polynomials. In the case of two
variables the problem at study is similar to Problem~\ref{prob:main}.
In this line of research the coefficients of the algebraic polynomial
are replaced with sufficiently refined intervals, hence obtaining
upper and lower bounds (i.e. a sleeve) for the polynomial.  Isolation
is performed using evaluations and exclusion predicates that involve
the non-vanishing of the derivative.  To our knowledge there is no
complexity analysis of the algorithms.  Neverthelss in
\cite{cgy-jsc-2009} evaluation bounds are presented, which are crucial
for the termination of the algorithm, based on separation bounds of
polynomial systems.  However, the systems used for the bounds involve
the derivative of the polynomial (this is needed for the exclusion
criterion), which is not the case for our approach.
In \cite{bclm-ascm-2009} the problem of real root isolation of 0-dim
square-free regular chains is considered. A generalization of
Vincent-Collins-Akritas (or Descartes) algorithm is used to isolate
the real roots of of polynomials with real algebraic numbers as
coefficients. This approach is similar to the direct strategy that we
study. To our knowledge the authors do not present a complexity
analysis since they focus on efficient algorithms and
implementation in \maple.

%\paragraph*{Our contribution} 
We revisit two approaches for isolating the real roots of a
square-free polynomial with coefficients in a simple algebraic
extension of the rational numbers. The first, indirect, approach
(Sec.~\ref{sec:RIC}), already presented in \cite{Johnson-phd-91},
is to find a polynomial with integer coefficients which is zero at all
roots of $B_{\alpha}$, isolate its real roots, and identify the
intervals which contain the roots of $B_{\alpha}$.  We compute
(aggregate) separation bounds for the resulting polynomial (Lem.~
\ref{lem:a-bounds}), that are slightly better than the ones in
\cite{rump-sigsam-1977}, and prove that the complexity of the
algorithm is $\sOB(N^{10})$, where $N$ is an upper bound on all the
quantities (degrees and bitsizes) of the input.
The second approach (Sec.~\ref{sec:Ba-bound}) is to isolate the roots
of the input polynomial directly, using either the Sturm's algorithm
or Sagraloff's modified Descartes algorithm. We analyze the worst-case
asymptotic complexity of the algorithms and we obtained a bound of
$\sOB(N^8)$ and $\sOB(N^7)$, respectively. We obtain these complexity
bounds by estimating improved separation bounds for the roots
(Sec.~\ref{sec:Ba-bound} and Lem.~\ref{lem:Ba-bounds}), that we also
prove that they are optimal (Sec.~\ref{Sec:Mignotte}).  The bounds are
better than the previously known ones \cite{rump76:algnum,Johnson-phd-91}
by a factor of $N$.  We empirically compare the performance of
the indirect approach and the direct approach based on Sagraloff's
modified Descartes algorithm.  The algorithms were implemented in
\func{C} as part of the core library of \mathematica, and we
illustrate their behavior on various datasets
(Sec.~\ref{sec:implementation}).
The complexity bounds that we present are many factors better that the
previously known ones. However, a fair and explicit comparison with
the bounds in \cite{Johnson-phd-91} is rather difficult, if possible
at all, since, besides the improved separation bounds that we present,
the complexity bounds of many sub-algorithms that are used have been
dramatically improved over the last 20 years, and it is not clear how
to take this into account in the comparison.

Finally, we present a generalization of the first approach to the
case where the input polynomials are univariate, but with coefficients
that belong to multiple extensions (Sec.~\ref{sec:multi}). We derive
(aggregate) separation bounds for this case
(Lem.~\ref{lem:a-multi-bounds}) and we sketch the overall complexity
of the algorithm. The bounds are single exponential with respect to
the number of extensions.

% The rest of the paper is structured as follows: First we introduce our
% notations, and in Sec.~\ref{sec:prelim} we present some
% preliminaries and known results that we will use throughout the
% paper.  In Sec.~\ref{sec:RIC} we present our first, indirect,
% approach for tackling Problem~\ref{prob:main} and in
% Sec.~\ref{sec:direct} the two direct algorithms.  In
% Sec.~\ref{sec:implementation} we present our implementation and
% experiments. Finally, in Sec.~\ref{sec:multi} we present the
% generalization of the first approach to the multiple extension case.

\paragraph*{Notation}
\OB means bit complexity and the \sOB-notation means that we
are ignoring logarithmic factors.
% We consider square-free polynomials.
For $A = \sum_{i=1}^{d}{a_i x^i} \in \ZZ[x]$, $\dg{A}$ denotes its
degree.  \bitsize{A} denotes an upper bound on the bitsize of the
coefficients of $A$, including a bit for the sign.  For $\rat{a} \in
\QQ$, $\bitsize{ \rat{a}}\ge 1$ is the maximum bitsize of the
numerator and the denominator.

If $\alpha_1, \dots, \alpha_d$ are the distinct, possible complex,
roots  of $A$, then 
$\Delta_i= |\alpha_i - \alpha_{c_i}|$, where 
$\alpha_{c_i}$ is the roots closest to $\alpha_{i}$.
$\Delta(A) = \min_i\Delta_i(A)$ is the separation bound,
that is the smallest distance between two (real or complex,
depending on the context) roots of $A$.
By $\Sigma(A) = - \sum_{i=1}^{n}{\lg{\Delta_i(A)}}$,
we denote the numbers of bits needed to represent
isolating rational numbers for all the roots of $A$.

Given two polynomials, possible multivariate, $f$ and $g$, then
$\res_x(f, g)$ denotes their resultant with respect to $x$.

\section{Preliminaries}
\label{sec:prelim}

Real algebraic numbers are the real roots of univariate polynomials
with integer coefficients; let their set be $\ALG$.  We
represent them in the so-called {\em isolating interval
  representation}. %, e.g.~\cite{Yap2000,BPR06}.  
If $\alpha \in \ALG$
then the representation consists of a square-free polynomial with
integer coefficients, $A \in\ZZ[x]$, that has $\alpha$ as a real root,
and an isolating interval with rational endpoints, $\interval{I} =
[\rat{a}_1, \rat{a}_2]$, that contains $\alpha$ and no other root of
the polynomial.  We write $\alpha \cong (A, \interval{I})$.  
%Notice that the representation is not unique.
%%
%For another type of representation of the elements of $\ALG$, Thom's
%encoding, we refer the reader to \cite{BPR06}.

%In the sequel we present several results that we use throughout the paper.

The following proposition provides various bounds for the roots of a
univariate polynomial.  Various versions of the proposition could be
found in
e.g.~\cite{Yap:SturmBound:05,Dav:TR:85,te-tcs-2008}.
We should mention that the constants that appear are not optimal.
For multivariate bounds we refer to \cite{emt-issac-2010}.

\begin{proposition}
  \label{prop:sep-bounds}
  Let $f$ be a univariate polynomial of degree $p$. If $\gamma_i$
  are the distinct real roots of $f$, then it holds 
  \begin{eqnarray}
      |\gamma_i| &\leq & 2\norm{f}_{\infty} \leq 2^{\tau+1} 
      \label{eq:u-upper} \enspace, \\
      -\lg{ \Delta(f) } &\leq& 
      -\frac{1}{2} \lg|3 \,\disc(f_{red})| + \frac{p+2}{2}\lg(p) + \nonumber\\ && (p-1)\lg\norm{f_{red}}_2
       \label{eq:u-sep}  \\
       & \leq &  2p\lg{p} + p\tau  \nonumber \enspace, \\
      -\sum_{i}\lg{ \Delta_i(f) } &\leq & 
      -\frac{1}{2} \lg|\disc(f_{red})| + \frac{p^2-p-2}{2}+ \nonumber\\&& (2p-1)\lg\norm{f_{red}}_2 
      \\ \nonumber
      &\leq & 3p^2 + 3p\tau + 4p\lg{p} \label{eq:u-dmm} \enspace,
  \end{eqnarray}
  where $f_{red}$ is the square-free part of $f$, and the second
  inequalities hold if we consider $f \in \ZZ[x]$ and $\bitsize{f} = \tau$.
\end{proposition}

\begin{proposition} 
  \label{pr:solve-1} 
  Let $f\in\ZZ[x]$ have degree $p$ and bitsize $\tau$.  We compute the
  isolating interval representation of its real roots and their
  multiplicities in $\sOB( p^5 + p^4 \tau + p^3 \tau^2)$
  \cite{t-arxiv-icf-10,s-arxiv-isol-10}.  The endpoints of the
  isolating intervals have bitsize $\OO( p^2 + p \, \tau)$ and
  $\bitsize{ f_{red}} = \OO( p + \tau)$, where $f_{red}$ is the
  square-free part of $f$.
  If $N = \max\{p, \tau\}$ then complexity bound for isolation becomes $\sOB(N^5)$.
\end{proposition}

\begin{proposition} \label{prop:sign-at-1} \cite{det-jsc-2009,emt-lncs-2006}
  Given a real algebraic number $\alpha \cong (f, [\rat{a}, \rat{b}])$, 
  where $\bitsize{\rat{a}} = \bitsize{ \rat{b} } = \OO( p^2 + p \tau)$,
  and $g \in \ZZ[x]$, such that $\deg(g) = q$, $\bitsize{ g} = \sigma$,
  we compute $\sign( g (\alpha))$
  in bit complexity $\sOB( pq \max\{\tau, \sigma\} + p \min\{p, q\}^2 \tau )$.
\end{proposition}

For the proofs of the following results the reader may refer to
\cite{det-jsc-2009}.  Let $f, g \in (\ZZ[x])[y]$ such that $\deg_x( f)
= p$, $\deg_x( g) = q$, $\deg_y( f), \deg_y( g) \leq d$, $\tau = \max(
\bitsize{ f}, \bitsize{ g})$.  By $\SR( f, g \,;\, \rat{a})$ we denote
the evaluation of the signed polynomial remainder sequence of $f$ and
$g$ with respect to $x$ over $\rat{a}$, and by $\SR_j( f, g \,;\,
\rat{a})$ the $j$-th element in this sequence.

\begin{proposition} 
  \label{prop:biv-SR-fast-computation}
  We can compute 
%  We compute $\SRQ( f, g)$, any polynomial in $\SR(f, g)$,
  $\res( f, g)$ w.r.t. $x$ or $y$ in 
  %% \\ $\sOB( \min\{p,q\} \max\{p,q\}^2 d \tau) = \sOB( p q \max\{p,q\} d \tau)$.
  $\sOB( p q \max\{p,q\} d \tau)$.
\end{proposition}

\begin{proposition} \label{prop:biv-SR-fast-evaluation}
  We compute $\SR( f, g \,;\, \rat{a})$, 
  where $\rat{a}\in\QQ \cup \{ \infty \}$ and $\bitsize{\rat{a}} = \sigma$,
  in
  %% \\$\sOB(\min\{p,q\} \max\{p,q\}^2 d \max\{ \tau, \sigma\}) = \sOB(p q \max\{p,q\} d \max\{ \tau, \sigma\})$.
  $\sOB(p q \max\{p,q\} d \max\{ \tau, \sigma\})$.
  For the polynomials $\SR_j( f, g \,;\, \rat{a}) \in \ZZ[y]$,
  except for $f, g$,
  we have $\deg_y( \SR_j( f, g \,;\, \rat{a})) = \OO( (p+q) d)$
  and $\bitsize{ \SR_j( f, g \,;\, \rat{a}) } = \OO( \max\{p,q\} \tau + \min\{p,q\} \sigma)$.
\end{proposition}

\section{Reduction to integer coefficients}
\label{sec:RIC}

\subsection{Some useful bounds}

The roots of $B_{\alpha}$ in Problem~\ref{prob:main} are algebraic numbers,
hence they are roots of a polynomial with integer coefficients. We
estimate bounds on the degree and the bitsize of this polynomial, and 
we will use them to analyze the Boolean complexity of the real
root isolation algorithm.
%We will use standard tools to derive the bounds.

Consider a real algebraic number $\alpha \in \ALG$,
in isolating interval representation $\alpha \cong ( A, \interval{I})$, where 
$A = \sum_{i=0}^{m}{ a_i\, x^i}$, $\interval {I} = [ \rat{a}_1 ,\rat{a_2}]$,
$\rat{ a}_{1,2} \in \QQ$ and $\dg{ A} = m$ and $\bitsize{ A} = \tau$.
Since $A$ is square-free, has $m$, possible complex, roots,
say $\alpha_1, \alpha_2, \dots, \alpha_m$ and after 
a (possible) reordering let $\alpha = \alpha_1$.

Let $B_{\alpha} \in \ZZ( \alpha)[ y]$, 
be a univariate polynomial in $y$, with coefficients that are polynomials
in $\alpha$ with integer coefficients.
More formally, let $B_{\alpha} = \sum_{i=0}^{n}{ b_i(\alpha) \, y^i}$,
where $b_i(x) = \sum_{j=0}^{\eta_i}{ c_{ij} \, x^j}$
and $\eta_i < m$, $0 \leq i \leq d$.
The restriction $\eta_i < m$ comes from the fact that $\ZZ(\alpha)$ is a
vector space of dimension%
\footnote{If $A$ is the minimal polynomial of $\alpha$ then the dimension is
  exactly $m$. In general it is not (computational) easy to compute
  the the minimal polynomial of a real algebraic number, thus we work
  with a square-free polynomial that has it as real root.}
$m$ and the elements of one of its bases are
$1, \alpha, \dots, \alpha^{m-1}$.
Finally, let $\bitsize{B_{\alpha}} = \max_{i,j}{ \bitsize{c_{ij}}} = \sigma$.
We assume that $B_{\alpha}$ is a square-free.

Our goal is to isolate the real roots of $B_{\alpha}$ (Problem~\ref{prob:main}). 
Since $B_{\alpha}$ has algebraic numbers as coefficients, its roots are
algebraic numbers as well. %, e.g. \cite{w-ma-53}.
Hence, there is a polynomial with integer coefficients that has as
roots the roots of $B_{\alpha}$, and possible other roots as well.
To construct this polynomial, e.g.~\cite{Johnson-phd-91,Dav:TR:85,l-cae-83},
we consider the following resultant w.r.t. $x$
\begin{equation}
  R( y)  =  \res_x( B(x, y), A( x) ) 
  =  (-1)^{m\eta} \, a_m^{\eta} \, \prod_{j=1}^{m}{ B( \alpha_j, y)},
  \label{eq:R-poly}
\end{equation}
where $\eta = \max\{\eta_i\}$, and $B(x,y) \in \ZZ[x, y]$ is obtained
from $B_{\alpha}$ after replacing all the occurrences of $\alpha$ with
$x$.  Interpreting the resultant using the Poisson formula, $R(y)$ is
the product of polynomials $B(\alpha_j, y)$, where $j$ ranges
over all the roots of $A$.  Our polynomial $B_{\alpha} \in \ZZ(
\alpha)[ y]$ is the factor in this product for $j = 1$.
Hence, $R$ has all the roots that $B_{\alpha}$ has and maybe
more.

\begin{remark}
  Notice that $R( y)$ is not square-free in general. For example
  consider the polynomial $B_{\alpha} = y^4 - \alpha^2$, where $\alpha$ is the
  positive root of $A = x^2 - 3$. In this case
  $R(y) = \res_x( A(x), B(x, y) = \res_x(x^2-3, y^2-x^2 ) = (y^4 -3)^2$.
\end{remark}

% \begin{remark}
%   If $A$ is irreducible, then to compute the minimal
%   polynomial of $B_{\alpha}$ it suffices to compute the square-free
%   factorization of $R$, using a result by Trager \cite{t-issac-76}.
% \end{remark}

Using Prop.~\ref{prop:sbiv-res-bounds} and  by taking into account
that $\eta_i < m$, we get 
$\deg(R) \leq mn$ and $\bitsize{ R} \leq m(\tau + \sigma)  + 2m\lg(4mn)$.
We may also write $\deg(R) = \OO( m n)$ and 
$\bitsize{ R} = \sO( m(\sigma + \tau))$.

In order to construct an isolating interval representation for the
real roots of $B_{\alpha}$, we need a square-free polynomial. This polynomial,
$C(y) \in \ZZ[y]$, is a square factor of $R(y)$, and so it holds
$\dg{C} \leq mn$ and $\bitsize{C} \leq m(\tau + \sigma) +
3m\lg(4mn)$, where the last inequality follows from Mignotte's bound \cite{Mign91}.

Using the Prop.~\ref{prop:sep-bounds}, 
we deduce the following lemma:
\begin{lemma}
  \label{lem:a-bounds}
  Let $B_{\alpha}$ be as in Problem~\ref{prob:main}.  The minimal
  polynomial, $C \in \ZZ[x]$, of the, possible complex, roots of
  $B_{\alpha}$, $\gamma_i$, has degree $\leq mn$ and bitsize $\leq
  m(\tau +\sigma) + 3m\lg(4mn))$ or $\sO(m(\tau +\sigma))$.
  Moreover, it holds 
  \begin{eqnarray}
    |\gamma_i| &\leq & 2^{m(\tau +\sigma) + 2m\lg(4mn)}  \enspace,\\ \label{eq:a-upper}
    -\lg{ \Delta(C) } &\leq& m^2n(\tau +\sigma + 4\lg(4mn)) 
    \enspace, \\ \label{eq:a-sep}
    %\Sigma(C) = 
    -\sum_{i}\lg{ \Delta_i(C) } &\leq & 3m^2n(n+\tau +\sigma + 6\lg(4mn))
    \enspace, \label{eq:a-dmm}
  \end{eqnarray}
%  or
  \begin{eqnarray}
    |\gamma_i| & \leq & 2^{\sO(m(\tau + \sigma))} \enspace,\\ \label{eq:a-upper-asympt}
    -\lg{ \Delta(C) } & = & \sO(m^2 n (\tau + \sigma)) \enspace,\\ \label{eq:a-sep-asympt}
    \Sigma(C) = -\sum_{i}\lg{ \Delta_i(C) } &= & \sO(m^2 n (n + \tau + \sigma))\enspace. \label{eq:a-dmm-asympt}
  \end{eqnarray}
\end{lemma}

\subsection{\label{Sec:ICF}The algorithm}

The indirect algorithm for Problem~\ref{pr:solve-1}, follows
closely the procedure described in the previous section to estimate
the various bounds on the roots of $B_{\alpha}$.
First, we compute the univariate polynomial with integer coefficients,
$R$, such that the set of its real roots includes those of $B_{\alpha}$.
We isolate the real roots of $R$ and we identify which ones are
roots of $B_{\alpha}$. 

Let us present in details the three steps and their complexity.
We compute $R$ using resultant computation, as presented in
(\ref{eq:R-poly}).  For this we consider $B$ as a bivariate polynomial
in $\ZZ[x, y]$ and we compute $\res_x( B(x, y), A(x))$, using
Prop.~\ref{prop:biv-SR-fast-computation}.  Since $\deg_x( B) < m$,
$\deg_y( B) = n$, $\bitsize{ B} = \sigma$, $\deg_x( A) = m$, $\deg_y(
A) = 0$ and $\bitsize{ A} = \tau$, this computation costs $\sOB( m^3 n
(\sigma + \tau))$, using Prop.~\ref{prop:biv-SR-fast-computation}.

Now we isolate the real roots of $R$.  This can be done in
 $\sOB( m^4 n^3 ( mn^2 +mn\tau +n^2\sigma +m\tau^2 + m\sigma^2 + m\tau \sigma))$, 
by Prop.~\ref{pr:solve-1}. 
In the same complexity bound we can also compute the multiplicities
of the real roots, if needed \cite{emt-lncs-2006}.

The rational numbers that isolate the real roots of $R$ have bitsize
bounded by $\sO(m^2 n (n+\sigma + \tau))$, which is also a bound on
the bitsize of all of them, as Prop.~\ref{prop:sep-bounds} and
Lem.~\ref{lem:a-bounds} indicate.

It is possible that $R$ can have more roots that $B_{\alpha}$, thus it
remains to identify which real roots of $R$ are roots of $B_{\alpha}$.
For sure all the real roots of $B_{\alpha}$ are roots of $R$.
Consider a real root $\gamma$ of $R$ and its isolating interval $[
\rat{c}_1, \rat{c}_2]$. If $\gamma$ is a root of $B_{\alpha}$, then
since $B_{\alpha}$ is square-free, by Rolle's theorem it must change
signs if we evaluate it over the endpoints of the isolating interval
of $\gamma$. Hence, in order to identify the real roots of $R$ that
are roots of $B_{\alpha}$ it suffices to compute the sign of
$B_{\alpha}$ over all the endpoints of the isolating intervals.  

We can improve the step that avoids the non-relevant roots of $R$ by
applying the algorithm for chainging the ordering of a bivariate
regular chain \cite{cs-issac-2006}. However, currently, this step is
not the bottleneck of the algorithm so we do not elaborate further.

Consider an isolating point of $R$, say $\rat{c_j} \in \QQ$, of
bitsize $s_j$.  To compute the sign of the evaluation of $B_{\alpha}$
over it, we proceed as follows.  First we perform the substitution $y =
\rat{c}_j$, and after clearing denominators, we get a number in $\ZZ[
\alpha]$, for which we want to compute its sign.  This is equivalent
to consider the univariate polynomial $B(x,\rat{c}_j)$ and to compute
its sign if we evaluate it over the real algebraic number $\alpha$.
We have $\deg( B(x, \rat{c}_j)) = \OO(m)$ and $\bitsize{ B( x,
  \rat{c}_j)} = \sO(\sigma + n s_j)$.  Hence the sign evaluation costs
$\sOB( m^3 \tau + m^2 \sigma + m^2 n s_j)$ using
Prop.~\ref{prop:sign-at-1}.  Summing up over all $s_j$'s, there are
$\OO( mn)$, and taking into account that $\sum_j{s_j} = \sO(m^2 n
(\sigma + \tau + n))$ (Lem.~\ref{lem:a-bounds}), we conclude that the
overall complexity of identifying the real roots of $B_{\alpha}$ is $\sOB( m^4
n^3 + m^4 n \tau + m^3 n \sigma + m^4 n^2 (\sigma + \tau))$.

The overall complexity of the algorithm is dominated by that of real
solving. We can state the following theorem:
\begin{theorem}
  The complexity of isolating the real roots of $B \in \ZZ( \alpha)[y]$
  using the indirect method is 
  $\sOB( m^4 n^3 ( mn^2 +mn\tau +n^2\sigma +m\tau^2 + m\sigma^2 + m\tau \sigma))$.
  If  $N = \max\{m, n, \sigma, \tau\}$, then the previous bounds
  become $\sOB( N^{10})$.
\end{theorem}

If the polynomial $B_{\alpha}$ is not square-free then we can apply
the algorithm of \cite{HoeMon:gcd:02} to compute its square-free
factorization and then we apply the previous algorithm either to the
square-free part or to each polynomial of the square-free
factorization.  The complexity of the square-free factorization is
$\sOB(m^2n(\sigma^2 + \tau^2) + mn^2(\sigma + \tau))$, and does not
dominate the aforementioned bound.

\section{Two direct approaches}
\label{sec:direct}

The computation of $R$, the polynomial with
integer coefficients that has the real roots of $B_{\alpha}$ is a
costly operation that we usually want to avoid.
If possible, we would like to try to solve the polynomial $B_{\alpha}$
directly, using one of the well-known subdivision algorithms,
for example \func{strum} or  \func{descartes} and \func{bernstein},
specially adopted to handle polynomials that have coefficients in an extension
field. 
In practice, this is accomplished by obtaining, repeatedly improved,
approximations of the real algebraic number $\alpha$ and subsequently
apply \func{descartes} or \func{bernstein} for polynomials with interval
coefficients, e.g.~\cite{RouZim:solve:03,jk-issac-1997}.

The fact that we compute the roots using directly the representation
of $B_{\alpha}$ allows us to avoid the complexity induced by the
conjugates of $\alpha$.  This leads to improved separation bounds, and
to faster algorithms for real root isolation.

\subsection{Separation bounds for $B_{\alpha}$}
\label{sec:Ba-bound}

We compute various bounds on the roots of $B_{\alpha}$ based on the
first inequalities of Prop.~\ref{prop:sep-bounds}.
For this we
need to compute a lower bound for $|\disc(B_{\alpha})|$ and an upper
bound for $\norm{B_{\alpha}}_2$.

First we compute bounds on the coefficients on $B_{\alpha}$.
Let $\alpha_{1}=\alpha,\alpha_{2},\ldots,\alpha_{m}$
be the roots of $A$.
We consider the resultants 
\begin{displaymath}
  r_{i}:=
  \res_{x}(A(x), z- b_i(x)) =
  \res_{x}\Paren{ A(x), z-\sum_{j=0}^{\eta_{i}}c_{i,j}x^{j}}
  \in\ZZ[z] \enspace. 
\end{displaymath}

It holds that 
\begin{displaymath}
  r_i(z) = a_m^{\eta} \prod_{k=1}^{m}(z- b_i(\alpha_k)) 
  \enspace,
\end{displaymath}
where $\eta =\max\{\eta_i\} < m$.
The roots of $r_i$ are the numbers $b_i(\alpha_k)$,
where $k$ runs over all the roots of $A$.
We use Prop.~\ref{prop:sbiv-res-bounds} to bound the degree and
bitsize of $r_i$.
The degree of $r_i$ is bounded by $m$ and their coefficient are of
bitsize $\leq m\sigma + m\tau + 5m\lg(m)$.
Using Cauchy's bound, we deduce
\begin{equation}
  2^{-m\sigma - m\tau - 5m\lg(m)} \leq \Abs{ b_i(\alpha_k)} \leq 2^{m\sigma + m\tau + 5m\lg(m)} 
  \enspace,
  \label{eq:bi-ineq}
\end{equation}
for all $i$ and $k$.
To bound $|\disc(B_{\alpha})|$ we consider the identity
\begin{displaymath}
  \begin{aligned}
    \disc(B_{\alpha}) &&=& (-1)^{\tfrac{1}{2}n(n-1)}\frac{1}{b_n(\alpha)}
    \res_y(B_{\alpha}, \partial B_{\alpha}(y) / \partial y) \\
    &&=& (-1)^{\tfrac{1}{2}n(n-1)}\frac{1}{b_n(\alpha)} \, R_{B}(\alpha)
    \enspace ,
  \end{aligned}
\end{displaymath}
where the resultant, $R_{B} \in \ZZ[\alpha]$, can be computed
as the determinant of the Sylvester matrix of $B_{\alpha}$ and
$\partial B_{\alpha}(y) / \partial y$, evaluated over $\alpha$.

The Sylvester matrix is of size $(2n-1) \times (2n-1)$, the elements
of which belong to $\ZZ[\alpha]$. 
The determinant consists of $(2n-1)!$ terms.
Each term is a product of $n-1$ polynomials in $\alpha$ of
degree at most $m-1$ and bitsize at most $\sigma$,
times a product of $n$ polynomials in $\alpha$ of degree at most $m-1$ and
bitsize at most $\sigma + \lg{n}$.
The first product results a polynomial of degree $(n-1)(m-1)$ and bitsize
$(n-1)\sigma + (n-1) \lg m$.
The second product results polynomials of degree $n(m-1)$ and bitsize
$n \sigma \lg n + n \lg{m}$.  Thus, any term in the determinant
expansion is a polynomial in $\alpha$ of degree at most $(2n-1)(m-1)$,
or $\OO(mn)$, and bitsize at most $4(2n-1)\sigma \lg(mn)$ or
$\sO(n\sigma)$.
The determinant itself, is a polynomial in $\alpha$ of degree at most 
$mn$ and of bitsize 
$4 (2n-1) \sigma \lg(mn) + (2n-1)\lg(2n-1)
\leq  5 (2n-1) \sigma \lg(mn) = \sO(n \sigma)$.

To compute a bound on $R_B(\alpha)$ we consider $R_B$ as a
polynomial in $\ZZ[y]$, and we compute a bound on its evaluation over
$\alpha$. For this we use resultants.  It holds
\begin{displaymath}
  D = \res_x( A(x), y-R_B(x)) = a_m^{\deg(R_B)} \prod_{i=1}^{m}(y-R_B(\alpha_i))
  \enspace .
\end{displaymath}
We notice that the roots of $D \in \ZZ[x]$ are the evaluations of $R_{B}$ over the roots
of $A$. So it suffices to compute bounds on the roots of $D$. 
Using Prop.~\ref{prop:sbiv-res-bounds} we deduce that 
$\deg(D) \leq m$ and 
$\bitsize{D} \leq 13 m n \sigma \lg(mn) + m n \tau$ or
$\bitsize{D} = \sO( mn(\sigma + \tau))$.
Using Cauchy bound, refer to Eq.~(\ref{eq:u-upper}), we conclude that
\begin{displaymath}
  2^{-13 m n \sigma \lg(mn) - m n\tau}
   \leq \abs{ R_B(\alpha)} \leq 
   2^{13 m n \sigma \lg(mn) + m n\tau}  \enspace .
\end{displaymath}
Using this inequality and (\ref{eq:bi-ineq}), 
we can bound $\abs{\disc(B_{\alpha})}$, i.e.
\begin{equation}
  2^{-13 m n \sigma \lg(mn) - 2m n\tau}
  \leq \abs{\disc(B_{\alpha})} \leq 
  2^{13 m n \sigma \lg(mn) + 2m n\tau} \enspace . 
  \label{eq:Ba-disc-bound}    
\end{equation}

It remains to bound $\norm{B_{\alpha}}_2$.
Using Eq.~(\ref{eq:bi-ineq}) we get
\begin{displaymath}
  \norm{B_{\alpha}}_2^2 \leq \sum_{i=0}^{n}(b_i(\alpha))^2
  \leq (n+1) \,2^{2m(\sigma + \tau + 5\lg(m))}
  \enspace .
\end{displaymath}

The previous discussion leads to the following lemma
\begin{lemma}
  \label{lem:Ba-bounds}
  Let $B_{\alpha}$ be as in Problem~\ref{prob:main}, 
  and $\xi_i$ be its roots. Then, it holds 
  \begin{eqnarray}
    |\xi_i| &\leq & 2^{m(\tau +\sigma + 5\lg{m})}  \enspace,\\ \label{eq:Ba-upper}
    -\lg{ \Delta(B_{\alpha}) } &\leq& 12mn(\sigma \lg(mn) + \tau  + 5\lg{m})  
    \enspace, \\ \label{eq:Ba-sep}
    %\Sigma(B_{\alpha}) = 
    -\sum_{i}\lg{ \Delta_i(B_{\alpha}) } &\leq & 
    14mn(\sigma \lg(mn) + \tau + 5\lg{m})
    \enspace, \label{eq:Ba-dmm}
  \end{eqnarray}
  or
  \begin{eqnarray}
    |\xi_i| & \leq & 2^{\sO(m(\tau + \sigma))} \enspace,\\ \label{eq:Ba-upper-asympt}
    -\lg{ \Delta(B_{\alpha}) } & = & \sO(m n (\tau + \sigma)) \enspace,\\ \label{eq:Ba-sep-asympt}
    \Sigma(B_{\alpha}) = -\sum_{i}\lg{ \Delta_i(B_{\alpha}) } &= & \sO(m n (\tau + \sigma))\enspace. \label{eq:Ba-dmm-asympt}
  \end{eqnarray}
\end{lemma}

\subsection{The {\sc sturm} algorithm}
\label{sec:sturm}

Let us first study the \func{sturm} algorithm. 
%It is a purely symbolic algorithm.
We assume $B_{\alpha}$ as in Problem~\ref{prob:main} to be square-free.
To isolate the real roots of $B_{\alpha}$ using the 
\func{sturm} algorithm, we need to evaluate the Sturm sequence of $B(\alpha,y)$
and its derivative with respect to $y$, $\partial B(\alpha,y) / \partial
y$, over various rational numbers.
For the various bounds needed we will use Lem.~\ref{lem:Ba-bounds}.
%Recall that the polynomial $R_{red}$ of the previous section {\em contains}
%all the real roots of $B$. Hence, we will use this polynomial in order
%to derive all the bounds that we need in order to perform our analysis.
%Recall that $\deg( R_{red}) = \OO( mn)$ and $\bitsize{ R} = \sO( mn + m(\sigma + \tau))$.

The number of steps that a subdivision-based algorithm, and hence
\func{sturm} algorithm, performs to isolate the real roots of a
polynomial depends on the separation bound. 
To be more specific, the number of steps, $(\#T)$, that \func{sturm} performs is 
$(\#T) \leq 2 r + r \lg{ \rat{ B}} + \Sigma(B_{\alpha})$ 
\cite{Dav:TR:85,Yap:SturmBound:05},
where $r$ is the number of real roots and $\rat{B}$ is an upper bound on the real roots.
Using (\ref{eq:Ba-upper}) and (\ref{eq:Ba-dmm}) we deduce that 
$(\#T) = \sO(mn(\tau +\sigma))$.
% \begin{displaymath}
%   (\#T) = \sO(mn(\tau +\sigma)) \enspace .
% \end{displaymath}

To complete the analysis of the algorithm it remains to compute the
complexity of each step, i.e. the cost of evaluating the Sturm
sequence over a rational number, of the worst possible bitsize.
The latter is induced by the separation bound, and in our case is 
$\sO(mn(\tau + \sigma))$.

We consider $B$ as polynomial in $\ZZ[x, y]$ and we evaluate the
Sturm-Habicht sequence of $B$ and $\frac{ \partial B}{ \partial y}$,
over rational numbers of bitsize $\sO(mn(\tau + \sigma))$.
The cost of this operation is $\sOB(m^2n^4(\tau+ \sigma))$ 
(Prop.~\ref{prop:biv-SR-fast-evaluation}).

It produces $\OO( n)$ polynomials in $\ZZ[x]$,
of degrees $\OO(mn)$ and bitsize $\sO(n \tau + n \sigma)$.
For each polynomial we have to compute its sign if we evaluate it over
$\alpha$.
Using Prop.~\ref{prop:sign-at-1} each sign evaluation costs
$\sOB( m(m^2 +n^2)\tau + mn^2\sigma)$, and so the overall cost is 
$\sOB( mn(m^2 +n^2)\tau + mn^3\sigma)$.
If we multiply the latter bound with the number of steps, 
$\sO(mn(\tau +\sigma))$, we get the following theorem.

\begin{theorem}
  The complexity of isolating the real roots of $B \in \ZZ( \alpha)[y]$
  using the \func{sturm} algorithm is 
  %$\sOB( m^2n^2(m^2 +n^2)\tau^2 + m^2n^2(m^2 +n^2)\sigma^2)$
  $\sOB( m^2n^2(m^2 +n^2)(\tau^2 +\sigma^2))$,
  or $\sOB( N^{8})$, where $N = \max\{m, n, \sigma, \tau\}$.
\end{theorem}

\subsection{\label{Sec:BMD}A modified {\sc descartes} algorithm}

We consider  Sagraloff's modified version of Descartes' algorithm
\cite{s-arxiv-isol-10}, that applies to polynomials with bitstream coefficients.
We also refer the reader to \cite{EigenEtal:bitstream:05,mr-jsc-2010}.

As stated in Problem~\ref{prob:main}, let $\alpha$ be a real root of
$A=\sum_{i=0}^{m}a_{i}x^{i}\in\mathbb{Z}[x]$, where $a_{m}\neq0$ and
$\abs{a_{i}} <2^{\tau}$ for $0\leq i\leq m$,
 and let
$B_{\alpha} = \sum_{i=0}^{n}b_{i}(\alpha) y^{i}\in\mathbb{Z}[\alpha][y]$, 
where
$b_{i}=\sum_{j=0}^{\eta_{i}}c_{i,j}x^{j} \in \ZZ[x]$ , $\eta_{i}<m$ and 
$\abs{c_{i,j}}<2^{\sigma}$ for $0\leq i\leq n$ and $0\leq j\leq
\eta_{i}$, where we also assume that $B_{\alpha}$ is square-free.

Let $\xi_{1},\ldots,\xi_{n}$ be all (complex) roots of $B$,
and $\Delta_{i}(B_{\alpha}) := \min_{j\neq i}\abs{\xi_{j}-\xi_{i}}$. 
By Theorem 19
of \cite{s-arxiv-isol-10}, the complexity of isolating real roots of
$B_{\alpha}$ is
\begin{displaymath}
 \sOB(n(\Sigma(B_{\alpha})+n\tau_{B})^{2}) \enspace,
\end{displaymath}
where $\Abs{\frac{b_{i}(\alpha)}{b_{n}(\alpha)}} \leq2^{\tau_{B}}$ and 
$\Sigma(B_{\alpha})= -\sum_{i=1}^{n}\lg(\Delta_i(B_{\alpha}))$.
From Lem.~\ref{lem:Ba-bounds} we get that 
\begin{equation}
  \Sigma(B_{\alpha}) \leq 14mn(\tau + \sigma \lg(mn)) +n\lg{n} = \sO(m n(\tau +\sigma))
  \enspace.
  \label{eq:sB-bd}
\end{equation}

To compute a bound on $\tau_{B}$, we use Eq.~(\ref{eq:bi-ineq}). 
It holds 
$\Abs{ \frac{b_i(\alpha_k)}{b_n(\alpha_k)}} \leq 2^{2m\sigma + 2m\tau + 6m\lg(m)}$,
for all $i$ and $k$.
Hence, 
\begin{equation}
  \tau_B \leq 2m\sigma + 2m\tau + 6m\lg(m) = \sO(m(\sigma + \tau))\enspace.
  \label{eq:tb-bd}
\end{equation}

Finally, by combining (\ref{eq:sB-bd}) and (\ref{eq:tb-bd}), we deduce
that the cost of isolating real roots of $B$ is
\begin{eqnarray*}
  \begin{aligned}
    \sOB(n(\Sigma(B_{\alpha})+n\tau_{B})^{2}) & = \sOB(n(mn\tau + mn\sigma)^2) \\
    %& =\sOB(n(m^2n^2\tau^2 + m^2n^2\sigma^2) ) \\
    & = \sOB(m^2n^3(\sigma^2+\tau^2))   \enspace.
  \end{aligned}
\end{eqnarray*}
If $N=\max\{m,n,\sigma,\tau\}$, then the bound becomes $\sOB(N^7)$.

It remains to estimate the cost of computing the successive
approximations of ${b_{i}(\alpha)}/{b_{n}(\alpha)}$.
The root isolation algorithm requires approximations of
${b_{i}(\alpha)}/{b_{n}(\alpha)}$ to accuracy of $\OO(\Sigma(B_{\alpha})+n\tau_{B})$ bits
after the binary point. Since 
$\abs{b_{i}(\alpha)/b_{n}(\alpha)} \leq 2^{\tau_{B}}$, 
to approximate each fraction, for $0\leq i\leq n-1$, 
to accuracy $L$, it is sufficient to approximate $b_{i}(\alpha)$,
for $0\leq i\leq n$, up to precision $\OO(L+\tau_{B})$. Hence, the
algorithm requires approximation of $b_{i}(\alpha)$, for $0\leq i\leq n$, to
precision $\OO(\Sigma(B)+n\tau_{B})$.  By inequality (\ref{eq:bi-ineq}),
$\abs{b_{i}(\alpha)} \geq 2^{-\tau_{B}}$, and therefore it is sufficient to
approximate $b_{i}(\alpha)$ to accuracy $\OO(\Sigma(B_{\alpha})+n\tau_{B})$.

Approximation of $c_{i,j}\alpha^{j}$ to accuracy of $L$ bits requires
approximation of $\alpha$ to accuracy of 
$ L+\lg|c_{i,j}|+\lg(j)+(j-1)\lg|\alpha| \leq
L+\sigma+\lg(m)+(m-1)(\tau+1) = \sO(L+\sigma+m\tau)
$
% \begin{displaymath}
%   \begin{aligned}
%     L+\lg|c_{i,j}|+\lg(j)+(j-1)\lg|\alpha| &\leq
%     L+\sigma+\lg(m)+(m-1)(\tau+1) \\&= \sO(L+\sigma+m\tau)
%   \end{aligned}
% \end{displaymath}
bits. Hence the accuracy of approximations of $\alpha$ required by
the algorithm is 
 \begin{displaymath}
   \OO(\Sigma(B_{\alpha})+n\tau_{B})= \sO(m n(\sigma+\tau)) \enspace .
 \end{displaymath}

By Lemmata 4.4, 4.5 and 4.11 of \cite{k-casc-09}, the bit complexity of approximating
$\alpha$ to accuracy $L$ is 
\begin{displaymath}
  \sO(m^{4}\tau^{2}+m^{2}L) \enspace .
\end{displaymath}
Therefore, the bit complexity of computing the required approximations
of ${b_{i}(\alpha)}/{b_{n}(\alpha)}$ is
\begin{displaymath}
  \sO(m^4\tau^2 + m^2 m n(\sigma+\tau)  )=
  \sO(m^3 (m \tau^2 + n \sigma + n \tau)) \enspace. 
\end{displaymath}
% and is dominated. by the complexity of isolating real roots of $B$.

\begin{theorem}
  The bit complexity of isolating the real roots of $B_{\alpha}$ of
  Problem~\ref{prob:main} using the modified Descartes' algorithm in
  \cite{s-arxiv-isol-10} is 
  $\sOB(m^2n^3(\sigma^2+\tau^2) + m^3 (m \tau^2 + n \sigma + n \tau))$, or
  $\sOB(N^7)$, where $N=\max\{m,n,\sigma,\tau\}$.
\end{theorem}

% \newpage

\subsection{\label{Sec:Mignotte}Almost tight separation bounds}

Let $\alpha$ be the root of
$A(x) = x^m - a x^{m-1} - 1$,
in $(a, a+1)$, for $a \geq 3$, $m \geq 3$. 
Then the Mignotte polynomial
$ B_{\alpha}(y) = y^n - 2(\alpha^k y-1)^2 $,
where $k=\floor{(m-1)/2}$,
has two roots in $(1/\alpha^k-h, 1/\alpha^k+h)$, where
$h = \alpha^{-k(n+2)/2} < a^{-(m-2)(n+2)/4} $.

If $a \leq 2^{\tau}$ and $\tau = \Omega( \lg(mn))$, then
$-\lg{\Delta(B_{\alpha})} = \Omega(mn\tau)$, which matches the upper
bound in (\ref{eq:Ba-sep}) of Lem.~\ref{lem:Ba-bounds}.  This
quantity, $\Omega(mn\tau)$, is also a tight lower bound for the number
of steps that an subdivision based algorithm performs, following the
arguments used in \cite{ESY:descartes} to prove a similar bound for
polynomials with integer coefficients.

\section{Implementation and experiments}
\label{sec:implementation}

We compare implementations of two methods of real root isolation for
square-free polynomials over simple algebraic extensions of rationals.
The first method, \emph{ICF} (for Integer Continued Fractions), performs
reduction to integer coefficients described in Section \ref{Sec:ICF}.
For isolating roots of polynomials with integer coefficients it uses
the \mathematica implementation of the Continued Fractions algorithm
\cite{LANA-2005}. The second method, \emph{BMD} (for Bitstream Modified
Descartes), uses Sagraloff's modified version of Descartes' algorithm
(\cite{s-arxiv-isol-10}, see Section \ref{Sec:BMD}). 
The algorithm has been implemented in C as a part of the \mathematica
system. 

For the experiments we used a 64-bit Linux virtual machine with
a $3$ GHz Intel Core i7 processor and $6$ GB of RAM. The timings
are in given seconds. Computations that did not finish in $10$ hours
of CPU time are reported as $>36000$.

\paragraph*{Randomly generated polynomials}
For given values of $m$ and $n$ each instance was generated as follows.
First, univariate polynomials of degree $m$ with uniformly distributed
random $10$-bit integer coefficients were generated until an irreducible
polynomial which had real roots was obtained. A real root $r$ of
the polynomial was randomly selected as the extension generator. Finally,
a polynomial in $\mathbb{Z}[r,y]$ of degree $n$ in $y$ and degree
$m-1$ in $r$ with $10$-bit random integer coefficients was generated.
The results of the experiments are given in Table \ref{Random-polynomials}.
Each timing is an average for $10$ randomly generated problems. 

\begin{table}[t] 
\scriptsize
  \centering
\begin{tabular}{|c|c|c|c|c|c|c|}
\hline 
 $n$& Algorithm & $m=2$ & $m=3$ & $m=5$ & $m=10$ & $m=20$\tabularnewline
\hline
\hline 
$10$ & $ICF$ & $0.003$ & $0.006$ & $0.013$ & $0.082$ & $0.820$\tabularnewline
\cline{2-7} 
 & $BMD$ & $0.002$ & $0.002$ & $0.003$ & $0.006$ & $0.019$\tabularnewline
\hline 
$20$ & $ICF$ & $0.004$ & $0.010$ & $0.048$ & $1.49$ & $2.80$\tabularnewline
\cline{2-7} 
 & $BMD$ & $0.008$ & $0.008$ & $0.010$ & $0.017$ & $0.053$\tabularnewline
\hline 
$50$ & $ICF$ & $0.014$ & $0.044$ & $0.271$ & $8.29$ & $20.5$\tabularnewline
\cline{2-7} 
 & $BMD$ & $0.046$ & $0.050$ & $0.061$ & $0.079$ & $0.213$\tabularnewline
\hline 
$100$ & $ICF$ & $0.047$ & $0.173$ & $1.09$ & $33.1$ & $108$\tabularnewline
\cline{2-7} 
 & $BMD$ & $0.165$ & $0.206$ & $0.137$ & $0.246$ & $0.546$\tabularnewline
\hline
$200$ & $ICF$ & $0.144$ & $0.612$ & $4.90$ & $141$ & $626$\tabularnewline
\cline{2-7} 
 & $BMD$ & $0.746$ & $0.701$ & $1.00$ & $0.824$ & $1.55$\tabularnewline
\hline
\end{tabular}
\caption{\label{Random-polynomials}Randomly generated polynomials}
\end{table}

\paragraph*{Generalized Laguerre Polynomials}
This example compares the two root isolation methods for generalized
Laguerre polynomials $L_{n}^{\alpha}(x)$, where $\alpha$ was chosen
to be the smallest root of the Laguerre polynomial $L_{m}(x)$. Note
that $L_{n}^{\alpha}(x)$ has $n$ positive roots for any positive
$\alpha$ and $L_{m}(x)$ has $m$ positive roots, so this example
maximizes the number of real roots of both the input polynomial with
algebraic number coefficients and the polynomial with integer coefficients
obtained by \emph{ICF}. The results of the experiment are given in
Table \ref{Laguerre polynomials}.

\begin{table}[t]
\centering
\scriptsize
\begin{tabular}{|c|c|c|c|c|c|c|}
\hline 
$n$ & Algorithm & $m=2$ & $m=3$ & $m=5$ & $m=10$ & $m=20$\tabularnewline
\hline
\hline 
$10$ & $ICF$ & $0.011$ & $0.008$ & $0.032$ & $0.208$ & $1.75$\tabularnewline
\cline{2-7} 
 & $BMD$ & $0.007$ & $0.007$ & $0.009$ & $0.010$ & $0.015$\tabularnewline
\hline 
$20$ & $ICF$ & $0.019$ & $0.041$ & $0.193$ & $1.50$ & $13.9$\tabularnewline
\cline{2-7} 
 & $BMD$ & $0.075$ & $0.071$ & $0.080$ & $0.088$ & $0.106$\tabularnewline
\hline 
$50$ & $ICF$ & $0.122$ & $0.270$ & $1.51$ & $25.8$ & $338$\tabularnewline
\cline{2-7} 
 & $BMD$ & $1.78$ & $1.63$ & $1.83$ & $1.90$ & $2.27$\tabularnewline
\hline 
$100$ & $ICF$ & $0.834$ & $2.17$ & $16.1$ & $365$ & $10649$\tabularnewline
\cline{2-7} 
 & $BMD$ & $54.7$ & $51.3$ & $56.0$ & $74.7$ & $92.4$\tabularnewline
\hline
$200$ & $ICF$ & $7.53$ & $31.2$ & $246$ & $8186$ & $>36000$\tabularnewline
\cline{2-7} 
 & $BMD$ & $2182$ & $3218$ & $3830$ & $4280$ & $4377$\tabularnewline
\hline
\end{tabular}
\caption{\label{Laguerre polynomials}Generalized Laguerre polynomials}
\end{table}

\paragraph*{Generalized Wilkinson Polynomials}
This example uses the following generalized Wilkinson polynomials
$W_{n,\alpha}(x):= \prod_{k=1}^{n}(x-k\alpha) , $
where $\alpha$ is the smallest root of the Laguerre polynomial $L_{m}(x)$.
The timings are presented in Table \ref{Wilkinson polynomials}. 

\begin{table}[t]
\centering
\scriptsize
\begin{tabular}{|c|c|c|c|c|c|c|}
\hline 
 $n$& Algorithm & $m=2$ & $m=3$ & $m=5$ & $m=10$ & $m=20$\tabularnewline
\hline
\hline 
$10$ & $ICF$ & $0.017$ & $0.012$ & $0.035$ & $0.285$ & $2.09$\tabularnewline
\cline{2-7} 
 & $BMD$ & $0.015$ & $0.013$ & $0.011$ & $0.015$ & $0.008$\tabularnewline
\hline 
$20$ & $ICF$ & $0.029$ & $0.069$ & $0.262$ & $2.23$ & $18.3$\tabularnewline
\cline{2-7} 
 & $BMD$ & $0.059$ & $0.052$ & $0.069$ & $0.039$ & $0.027$\tabularnewline
\hline 
$50$ & $ICF$ & $0.137$ & $0.356$ & $2.04$ & $45.4$ & $429$\tabularnewline
\cline{2-7} 
 & $BMD$ & $1.84$ & $1.35$ & $1.29$ & $0.703$ & $0.561$\tabularnewline
\hline 
$100$ & $ICF$ & $0.808$ & $2.84$ & $24.6$ & $674$ & $8039$\tabularnewline
\cline{2-7} 
 & $BMD$ & $47.0$ & $38.6$ & $32.0$ & $23.3$ & $8.38$\tabularnewline
\hline
$200$ & $ICF$ & $8.48$ & $35.1$ & $348$ & $11383$ & $>36000$\tabularnewline
\cline{2-7} 
 & $BMD$ & $3605$ & $2566$ & $2176$ & $927$ & $565$\tabularnewline
\hline
\end{tabular}
\caption{\label{Wilkinson polynomials}Generalized Wilkinson polynomials}
\end{table}

\paragraph*{Mignotte Polynomials}

The variant of Mignotte polynomials used in this example is given
by $M_{n,\alpha}(x):=y^{n}-2(\alpha^{k}y-1)^{2}$,
where $\alpha$ is the root of $A_{m}(x):=x^{m}-3x^{m-1}-1$
in $(3,4)$, $m\geq3$ and $k=\floor{(m-1)/2}$ (see Section \ref{Sec:Mignotte}).
The results of the experiment are given in Table \ref{Mignotte polynomials}. 
\begin{table}[t]
\centering
\scriptsize
\begin{tabular}{|c|c|c|c|c|c|}
\hline 
 $n$& Algorithm & $m=3$ & $m=5$ & $m=10$ & $m=20$\tabularnewline
\hline
\hline 
$10$ & $ICF$ & $0.003$ & $0.008$ & $0.049$ & $0.594$\tabularnewline
\cline{2-6} 
 & $BMD$ & $0.010$ & $0.006$ & $0.014$ & $0.036$\tabularnewline
\hline 
$20$ & $ICF$ & $0.006$ & $0.027$ & $0.288$ & $8.83$\tabularnewline
\cline{2-6} 
 & $BMD$ & $0.015$ & $0.020$ & $0.049$ & $0.137$\tabularnewline
\hline 
$50$ & $ICF$ & $0.041$ & $0.441$ & $12.2$ & $777$\tabularnewline
\cline{2-6} 
 & $BMD$ & $0.112$ & $0.147$ & $0.321$ & $0.854$\tabularnewline
\hline 
$100$ & $ICF$ & $0.866$ & $11.6$ & $729$ & $28255$\tabularnewline
\cline{2-6} 
 & $BMD$ & $0.702$ & $0.868$ & $2.32$ & $5.99$\tabularnewline
\hline
$200$ & $ICF$ & $35.7$ & $684$ & $23503$ & $>36000$\tabularnewline
\cline{2-6} 
 & $BMD$ & $3.12$ & $5.30$ & $13.8$ & $46.1$\tabularnewline
\hline
\end{tabular}
\caption{\label{Mignotte polynomials}Mignotte polynomials}
\end{table}

\vspace{5pt}

The experiments suggest that for low degree extensions \emph{ICF}
is faster than \emph{BMD}, but in all experiments as the degree of
extension grows \emph{BMD} becomes faster than \emph{ICF}. Another
fact worth noting is that \emph{ICF} depends directly on the extension
degree $m$, since it isolates roots of a polynomial of degree $mn$.
On the other hand, the only part of \emph{BMD} that depends directly
on $m$ is computing approximations of coefficients, which in practice
seems to take a very small proportion of the running time. The main
root isolation loop depends only on the geometry of roots, which depends
on $m$ only through the worst case lower bound on root separation.
Indeed, in all examples the running time of \emph{ICF} grows substantially
with $m$, but the running time of \emph{BMD} either grows at a much
slower pace or, in case of generalized Wilkinson polynomials, it even
decreases with $m$ (because the smallest root $\alpha$ of $L_{m}(x)$,
and hence the root separation of $W_{n,\alpha}(x)$, increase with $m$).
The superiority of the direct approach was also observed in \cite{Johnson-phd-91}.

%\section{Real root isolation in multiple extensions}
\section{Multiple extensions}
\label{sec:multi}

In this section we consider the problem of real root isolation of a
polynomials with coefficients in multiple extensions. We tackle the
problem using a reduction to a polynomial with integer
coefficients. The technique could be considered as a generalization of
the one presented in Sec.~\ref{sec:RIC}.

We use ${\x}^{\mathbf{ e}}$ to denote the monomial
$x_1^{e_1} \cdots x_n^{e_{\ell}}$, with
$\mathbf{ e} = (e_1, \dots, e_{\ell}) \in \NN^{\ell}$.
For a polynomial 
$f = \sum_{j=1}^{m}{ c_{j} {\x}^{\mathbf{e}_{j}}} \in \ZZ[\x]$,
let $\Set{ \mathbf{e}_{1}, \dots, \mathbf{e}_{m}} \subset \NN^{\ell}$ be the support of $f$;
its Newton polytope $Q$ is the convex hull of the support.
By $(\#Q)$ we denote the integer points of the polytope $Q$,
i.e. $(\#Q) = |Q \cap \ZZ^{\ell}|$.

%We consider the following problem, which generalizes
%Problem~\ref{prob:main} to multiple extensions.
%
\begin{framed}
\begin{problem}
  \label{prob:multi-ext}
  Let $\alpha_j$, where $1 \leq j \leq \ell$,
  be a real algebraic numbers. Their  isolating
  interval representation is $\alpha_j \cong ( A_j, \interval{I}_j)$, where
  $A_j=\sum_{i=0}^{m}{ a_i\, x_j^i}$, $\interval {I}_j = [ \rat{a}_{j,1}
  ,\rat{a_{j,2}}]$, $\rat{ a}_{1,2} \in \QQ$,
  $\dg{ A_j} = m$, and $\bitsize{ A_j} = \tau$.
  Let 
  \begin{displaymath}
    B_{\alpha} = 
    \sum_{i=0}^{n}{ b_i(\alpha_1, \dots, \alpha_{\ell}) \, y^i} 
    \in \ZZ( \alpha)[  y] ,
  \end{displaymath}  
  be square-free, where
  $b_i(\x) = \sum_{\mathbf{e}}{ c_{ij} \, \x^{\mathbf{e}}} \in \ZZ[\x]$,
  $\bitsize{c_{i,j}} \leq \sigma$,  for $0 \leq i \leq d$,
  and for  
  $\mathbf{ e} = (e_1, \dots, e_{\ell})$, it holds $e_j \leq \eta < m$, 
  What is the Boolean complexity of isolating the real roots of $B_{\alpha}$?
\end{problem}
\end{framed}

We denote by $\mathbf{a}_i$ the coefficients of $A_i$, where $1\leq i
\leq \ell$, and by $\mathbf{c}$ the coefficients of $B$. 
We compute separation bounds following the technique introduced 
\cite{emt-issac-2010}.%, see also \cite{c-crmp-87,Yap2000}.

We consider the zero dimensional polynomial system 
$(S): A_1(\x) = \cdots = A_{\ell}(\x) = A_{\ell+1}(\x) = 0$,
where $A_k(\x)  = \sum_{i=0}^{m}{ a_{k,i}\, x_k^i} = 0$, $1 \leq k \leq \ell$,
and $A_{\ell+1}=B(\x, y) = \sum_{i=0}^{n}{ b_i(x_1, \dots, x_{\ell}) \, y^i} =0$.
% \begin{displaymath}
%   (S) \quad
%   \left\{
%   \begin{aligned}
%     A_1(\x) & = \sum_{i=0}^{m}{ a_{1,i}\, x_1^i} = 0\\
%     & \vdots \\
%     A_{\ell}(\x) & = \sum_{i=0}^{m}{ a_{\ell,i}\, x_{\ell}^i} = 0\\
%     A_{\ell+1}=B(\x, y) & = \sum_{i=0}^{n}{ b_i(x_1, \dots, x_{\ell}) \, y^i} =0
%   \end{aligned}
%   \right.
% \end{displaymath}
We should mention that we make the assumption that $B$ does not become
identically zero when $\alpha_1, \dots,  \alpha_l$ are replaced with some set
of their conjugates (otherwise the resultant is zero). 

We hide variable $y$, that is we consider $(S)$ as an overdetermined
system of $\ell+1$ equations in $\ell$ variables.
We consider the resultant, $R$, with respect to $x_1, \dots, x_{\ell}$, that is
we eliminate these variables, and we obtain a polynomial 
$R \in \ZZ[\mathbf{a}_1, \dots, \mathbf{a}_{\ell}, \mathbf{c}, y ]$.
We interpret the resultant using the Poisson formula \cite{CLO2}, see also \cite{PeSt}, 
i.e.
\begin{displaymath}
  R(y) = \res_x( A_1, \dots, A_{\ell}, B) =
  \prod B( \alpha_{1,i_1}, \dots, \alpha_{\ell, i_{\ell}}, y)
  \enspace,
\end{displaymath}
and  $R(y) \in (\ZZ[\mathbf{a}_1, \dots, \mathbf{a}_{\ell}, \mathbf{c}])[y]$.
Similar to the single extension case, $B_{{\bm \alpha}}$, is among the
factors of $R$, hence it suffices to compute bounds for the roots of
$R(y)$.

We consider $R$ as a univariate polynomial in $y$.  The resultant is a
homogeneous polynomial in the coefficients of $(S)$, we refer to
e.g. \cite{CLO2,PeSt} for more details and to \cite{emt-issac-2010}
for a similar application. To be more specific, the structure of the
coefficients of $R$ is
\begin{displaymath}
  R(y) = \cdots + \varrho_k \, \mathbf{a}_1^{\M_1} \cdots  \mathbf{a}_{\ell}^{\M_{\ell}}
  \, \mathbf{c}^{\M_{\ell+1}-k} (y^i)^{k} + \cdots
  \enspace ,
\end{displaymath}
where $1 \leq k \leq \M_{\ell+1} = m^{\ell}$,
and $i$ is a number in $\set{1, \dots, n}$. The semantics of
$\mathbf{a}_i^{\M_i}$ are that it is a monomial in the coefficients of
$A_i$ of total degree $\M_i$.  Similarly, $\mathbf{c}^{\M_{\ell+1}-k}$
stands for a monomial in the coefficients of $B$ of total degree
$\M_{\ell+1}-k$.
Moreover, $\M_i \leq \ell \eta m^{\ell-1} < \ell(m-1)m^{\ell-1} < \ell m^{\ell}$.
The degree of $R$ with respect to $y$ is at most $n\,\M_{\ell+1}= n m^{\ell}$.

Since $\abs{a_{i,j}} \leq 2^{\tau}$,
it holds
\begin{equation}
  \lg \prod_{i=i}^{\ell} \abs{\mathbf{a}_i}^{\M_i} \leq \tau \ell^2 m^{\ell}
  \enspace .
  \label{eq:acoeff-bd}
\end{equation}

Similarly, since $\abs{c_{i,j}} \leq 2^{\sigma}$, we get
\begin{equation}
  \lg \abs{\mathbf{c}}^{\M_{\ell+1}-k} \leq \sigma (m^{\ell} - k) \leq
  \sigma m^{\ell} 
  \enspace .
  \label{eq:ccoeff-bd}
\end{equation}

Finally,
$\abs{\varrho_k} \leq \prod_{i=1}^{\ell+1}{(\#Q_{i})^{\M_i}}$  \cite{sombra-ajm-2004},
where $(\#Q_i)$ is the number of integer points of the Newton polytope
of the polynomial $A_i$. We let $A_{\ell+1} = B$.
It is $(\#Q_i) = m+1$ for $1 \leq i \leq \ell$, so
%\begin{displaymath}
$  \prod_{i=1}^{\ell}{(\#Q_{i})^{\M_i}} \leq (m+1)^{\ell (m-1) m^{\ell-1}}
  \leq  m^{\ell m^{\ell}}$ ,
%  \enspace,
%\end{displaymath}
and $(\#Q_{\ell+1}) \leq (\ell(m-1)+n)^{\ell+1} + \ell+1$.
Hence,
\begin{displaymath}
  \begin{aligned}
    (\#Q_i)^{\M_{\ell+1}} && \leq &\Paren{(\ell(m-1)+n)^{\ell+1} + \ell+1}^{m^{\ell}} \\
    &&\leq &\Paren{2 \ell m +n}^{(\ell+1)m^{\ell}}
    \leq \Paren{\ell m n}^{\ell m^{\ell}} \enspace ,
  \end{aligned}
\end{displaymath}
and so for every $k$
\begin{equation}
  \lg{\abs{\varrho_k}} \leq \lg {\prod_{i=1}^{\ell+1}{(\#Q_{i})^{\M_i}}}
  \leq {2\ell m^{\ell} \lg(mn\ell)}
  \enspace .
    \label{eq:rhocoeff-bd}
\end{equation}

By combining (\ref{eq:acoeff-bd}), (\ref{eq:ccoeff-bd}) and
(\ref{eq:rhocoeff-bd}) we can bound the coefficients of $R$ and its
square-free factors.
Using also Prop.~\ref{prop:sep-bounds} we get the following lemma.

\begin{lemma}
  \label{lem:a-multi-bounds}
  Let $B_{\alpha}$ be as in Problem~\ref{prob:multi-ext}.  The minimal
  polynomial, $C_{\ell}$ of the, possible complex, roots of $B_{\alpha}$,
  $\gamma_i$, has degree $\leq n\, m^{\ell}$ and bitsize 
  $\leq m^{\ell}(\tau \ell^2  + \sigma  + 3\ell\lg(mn\ell))$ or 
  $\sO(m^{\ell}(\ell^2\tau +\sigma))$.
  Moreover, it holds 
  \begin{eqnarray}
    |\gamma_i| &\leq & 2^{m^{\ell}( \ell^2\tau + \sigma + 2\ell\lg(mn\ell))}  \enspace,  \label{eq:am-upper} \\
    -\lg{ \Delta(C_{\ell}) } &\leq& 
    m^{2\ell}n( \ell^2 \tau + \sigma + 4\ell\lg(mn\ell))
    \enspace,  \label{eq:am-sep} \\
    %\Sigma(C_{\ell}) = 
    -\sum_{i}\lg{ \Delta_i(C_{\ell}) } &\leq & 
    m^{2\ell}n( \ell^2 \tau + \sigma + n + 6\ell\lg(mn\ell))
    \label{eq:am-dmm}
  \end{eqnarray}
%  or
  \begin{eqnarray}
    |\gamma_i| & \leq & 2^{\sO(m^{\ell}( \ell^2 \tau + \sigma))} 
    \enspace, \label{eq:am-upper-asympt} \\
    -\lg{ \Delta(C_{\ell}) } & = & \sO(m^{2\ell} n (\ell^2\tau +\sigma)) \enspace, \label{eq:am-sep-asympt} \\
    %\Sigma(C_{\ell}) =
    -\sum_{i}\lg{ \Delta_i(C_{\ell}) } &= & \sO(m^{2\ell} n(\ell^2\tau +\sigma +n))
    \enspace. \label{eq:am-dmm-asympt}
  \end{eqnarray}
\end{lemma}

\begin{remark}
  To match exaclty the bounds derived in Lem.~\ref{lem:a-bounds} one  should use
  for $\M_i$ the more accurate inequality $\M_i < \ell(m-1) m^{\ell-1}$.
\end{remark}

We can isolate the real roots of $C_{\ell}$ in 
$\sOB( n^5 m^{5\ell} +
 n^4 m^{5\ell}\tau\ell^2 +
 n^4 m^{4\ell} \sigma +
+n^3 m^{5\ell}\tau^2\ell^4 +
+n^3 m^{3\ell} \sigma^2)$.
%\cite{t-arxiv-icf-10,s-arxiv-isol-10}.  
That is we get a single
exponential bound with respect to the number of the real algebraic
numbers involved.

%\section{Conclusion and future work}
%\onehalfspacing

{
\subsection*{Acknowledgement}
ET is partially supported by an individual postdoctoral grant from the
Danish Agency for Science, Technology and Innovation,
and also acknowledges support from the Danish National Research Foundation and
the National Science Foundation of China (under the grant 61061130540)
for the Sino-Danish Center for the Theory of Interactive Computation,
within which part of this work was performed.
}

%\section{Conclusion and future work}

{ 
  %\scriptsize
  \bibliographystyle{abbrv}  
  \bibliography{rsef} 
  %\bibliography{Personal,et,algebra,complexity,emiris,arithm,numeric,geometry,soft,ecg,tmp} 
}

\newpage
\appendix

\section{A bound for the resultant}

\begin{proposition}
  \label{prop:sbiv-res-bounds}
  Let $B=\sum_{i,j}c_{i,j}x^{i} y^j \in\ZZ[x,y]$ of degree $n$ with
  respect to $y$ and of degree $\eta$ with respect to $x$, and of
  bitsize $\sigma$.  Let $A =\sum_{i=0}^{m}{a_ix^i} \in \ZZ[x]$ of
  degree $m$ and bitsize $\tau$. The resultant of $B$ and $A$ with
  respect to $x$ is univariate polynomial in $y$ of degree at most
  $mn$ and bitsize at most $m\sigma + \eta \tau +  m\lg(n+1) + (m+\eta)\lg(m +\eta)$ 
  or $\sO(m\sigma + \eta\tau)$.
\end{proposition}
\begin{proof}
  The proof follows closely the proof in
  \cite[Prop.~8.15]{BPR06} that provides a bound for general
  multivariate polynomials. 
  We can compute the resultant of $B(x, y)$ and $A(x)$ with respect to
  $x$ from the determinant of the Sylvester matrix, by
  considering them as univariate polynomial in $x$, with coefficients
  that are polynomial in $y$, which is 
  {%\scriptsize
  \begin{displaymath}
    %\Syl(B, A):=
    \left( 
      \begin{array}{ccccccc}
        b_\eta & b_{\eta-1} & \dots & b_0 &&&\\
        & b_\eta & b_{\eta-1} & \dots & b_0 &&\\
        && \ddots & \ddots && \ddots &\\
        &&& b_\eta & b_{\eta-1} & \dots & b_0 \\
        a_m & a_{m-1} & \dots & a_0 &&&\\
        & a_m & a_{m-1} & \dots & a_0 &&\\
        && \ddots & \ddots && \ddots &\\
        &&& a_m & a_{m-1} & \dots & a_0
      \end{array}
    \right) \,
    \begin{array}{c}
      x^{m-1} B \\
      x^{m-2} B \\
      \vdots \\
      x^0 B \\
      x^{\eta-1} A\\
      x^{\eta-2} A \\
      \vdots \\
      x^0 A \\
    \end{array}
  \end{displaymath}
  }
  where $b_k = \sum_{i=0}^{n}c_{i,k}y^i$.

  The resultant is a factor of the determinant of the Sylvester
  matrix. The matrix is of size $(\eta+m) \times (\eta + m)$, hence
  the determinant consists of $(\eta+m)!$ terms.  Each term is a
  product of $m$ univariate polynomials in $y$, of degree $n$ and
  bitsize $\sigma$, times the product of $n$ numbers, of bitsize
  $\tau$.  The first product results in polynomials in $y$ of degree
  at most $mn$ and bitsize at most $m\sigma + m\lg(n+1)$; since there
  are at most $(n+1)^m$ terms with bitsize at most $m\sigma$ each. The
  second product results in numbers of bitsize at most $\eta
  \tau$. Hence each term of the determinant is, in the worst case a
  univariate polynomial in $y$ of degree $m$ and bitsize $m\sigma +
  \eta \tau + m\lg(n+1)$.  We conclude that the resultant is of degree
  at most $mn$ in $y$ and of bitsize 
  $m\sigma + \eta \tau + m\lg(n+1) + (m+\eta)\lg(m +\eta)$ or $\sO(m\sigma + \eta\tau)$. 
\end{proof}

\end{document}